\documentclass[12pt]{article}
\usepackage[margin=1in]{geometry} 
\usepackage{amsmath,amsthm,amssymb,scrextend}
\usepackage{fancyhdr}
\usepackage{comment}
\usepackage{graphicx}
\usepackage{color}
\usepackage{cite}
\usepackage{authblk}

\newtheorem{theorem}{Theorem}
\newtheorem{lemma}{Lemma}

\newtheorem{corollary}{Corollary}
\newtheorem*{theorem1}{Theorem}
\newtheorem{notation}{Notation}

\newcommand{\cc}[0]{\hat{c}}
\newcommand{\rr}[0]{\hat{r}}

\begin{document}
 
\title{A Note on the Ratio of Revenues\\Between Selling in a Bundle and Separately}
\author{Ron Kupfer \thanks{School of Computer Science and Engineering, and Center for the Study of Rationality, Hebrew
University of Jerusalem.}}
\maketitle

\begin{abstract}
We consider the problem of maximizing revenue when selling $k$ items to a single buyer with known valuation distributions. Hart and Nisan \cite{hart2012approximate} raised some questions on the approximation of maximizing revenue from selling a set of items by selling each item separately or all in a bundle. In this work we show that for a single, additive buyer whose valuations for for the items are distributed according to i.i.d.\ distributions which are known to the seller, the  ratio of revenue from selling in a  bundle to selling separately is at least $55.9\%$ and this gap is attainable. 
\end{abstract}

\section{Introduction}
We are looking at the problem of maximizing revenue in the scenario of a single seller selling multiple items to a single buyer. This paper is limited to the case where the bidder's valuations of the items are identically and independently distributed and the bidder valuation is additive. The distributions (given by a cumulative distribution $F$) are known to the seller but their realizations are not.

For one item, the problem was completely solved by Myerson's classic result \cite{myerson1981optimal}. If the seller offers to sell it for a price $p$, then the probability that the buyer will buy is $1 - F(p)$, and the revenue will be $p\cdot(1 - F(p))$. The seller will choose a price $p^*$ that maximizes this expression. Myerson's characterization of optimal auctions concludes that the take-it-or-leave-it offer at the above price $p^*$ yields the optimal revenue among all mechanisms. Myerson's result also applies when there are multiple buyers, in which case $p^*$ would be the reserve price in a second price auction.

For more then one item, the problem of maximizing the revenue is more complicated. One approach is selling each item separately with take-it-or-leave-it offer of its Myerson's price. We denote the revenue from this approach by $SRev$. A second approach is considering all the items as one bundle and selling the bundle for its  Myerson's price. We denote the revenue from this approach by $BRev$.\\
Previous work by Hart and Nisan \cite{hart2012approximate} have shown that while the ratio $SRev/BRev$ may be arbitrarily small for some one-dimensional distributions, the opposite is not true. That is, the infimum of the ratio of $BRev/SRev$ over all the one-dimensional distributions is bounded.\\
Despite the fact that when the number of items goes to infinity this ratio goes to 1 (see, for example, Appendix E in \cite{hart2012approximate}), it was not clear what the exact ratio is for a finite number of items.
Hart and Nisan have shown this ratio is bounded between $\frac{1}{4}$ and $57\%$. We show a tight result on this proportion of approximately $55.9\%$. 

\begin{theorem} \label{main_theorem}
Let $\rr\approx55.9\%$\footnote{$\rr$ will be define formally in the next section}. For every integer $k\geq 1$ and every one-dimensional distribution $F$
\[BRev(F^{k}) \geq \rr kRev(F) = \rr SRev(F^{k})\] 
and the worst-case bound is tight in the sense that for every $\epsilon>0$, there exist an integer $k$ and a one-dimensional distribution $F$ such that 
\[BRev(F^{k}) < (\rr+\epsilon)SRev(F^{k})\]
\end{theorem} 

\section{Notation and Preliminaries}

For a cumulative distribution $F$ on \(\mathbb{R}^{k}_{+}\) (for \(k \geq 1\)), we consider the optimal revenue, $Rev(F)$, obtainable from selling $k$ items to a (single, additive) buyer whose valuation for the $k$ items is jointly distributed according to $F$, and define:
\begin{itemize}
\item \(SRev(F)\) is the maximal revenue obtainable by selling each item separately.
\item \(BRev(F)\) is the maximal revenue obtainable by bundling all items and selling them together.
\end{itemize}

Let $B(m;k,c/k) = \sum_{i = 0}^{m}\binom{k}{m}(\frac{c}{k})^i(1-\frac{c}{k})^{k-i}$ be the cumulative function of the binomial distribution\footnote{I.e., the probability of at most $m$, out of $k$, to be successful where each event is successful with independent probability $c/k$}
and \(P(m;c)=e^{-c}\sum_{i = 0}^{m}{c^i/i!}\) be the cumulative function of the Poisson distribution, which is a known approximation to $B(m;k,c/k)$ as $k$ goes to infinity.

In this work, we will use an approximation for the probability for a buyer to have a valuations sum of at least $d$ where each item valuation is distributed over $\{0,1\}$ with \(\mathbb{P}(1) = x/k\). This event is equivalent to the complementary event described by $B(d-1;k,x/k)$ and its probability is $1 - B(d-1;k,x/k)$.

\begin{notation}
	$h_d(x) = 1-P(d-1;x) = 1-e^{-x}\sum_{i = 0}^{d-1}{\frac{x^i}{i!}}$ is an approximation of the above probability.

In addition, we will need the following constants:
\begin{enumerate}
\item \(\cc = 2.655...\) is the solution $x$ of \( 2h_2(x) =  3h_3(x)\)
\item \(\rr = 2h_2(\cc)/\cc \approx 0.599...\)
\end{enumerate}
\end{notation}

\section{Upper Bound}

In this section we prove the upper bound of Theorem \ref{main_theorem} on the worst-case $BRev/SRev$ ratio.

\begin{theorem1} \label{my_examp}
For any $\epsilon>0$, there exist an integer $k$ and a one-dimensional distribution $F$ such that 
\[BRev(F^{k}) < (\rr+\epsilon)SRev(F^{k})\]
\end{theorem1}
 
\begin{proof}
Consider the distribution $F$ over \(\{ 0,1\}\) with \(\mathbb{P}(1) = \cc/k\), so the revenue from optimally selling a single item is $\cc/k$. For estimating the bundle revenue from posted price $d$, we can use the approximation $h_d(\cc)$ for having a valuations sum of at least $d$ and getting an expected revenue of $dh_d(\cc)$. Doing so, we can see that the revenue from posting bundle price of $2$ or $3$ will gain approximately $2h_2(\cc) =  3h_3(\cc) \approx 1.48...$ .\\
For large enough $k$, the probability of the valuations sum to be at least $d$ is bounded by $\cdot \cc^d/d!$. Hence, the revenue is bounded by $d\cdot \cc^d/d!$, which for $d>7$ is$\ <1.3$. For $4\leq d\leq 7$, we have $dh_d(\cc)<1.48$. Thus, \(BRev(F^{k})/SRev(F^{k}) = 2h_d(\cc)/\cc + \epsilon = \rr + \epsilon\) for $k$ large enough.
\end{proof}

\section{Lower Bound}
In this section we will show the lower bound of Theorem \ref{main_theorem}. In order to show this, we first prove the theorem in the case where the values are drawn from a Bernoulli distribution. We start by showing in Lemma \ref{c_geq_40} that in this case, the distribution $F$ over \(\{ 0,1\}\) with \(\mathbb{P}(1) = c/k\) with $c\geq 40$, will gain a high enough profit with high enough probability by selling in a bundle. Then we handle the case of $c<40$ by proving it separately for $k=2,3$, and $k \geq 4$ in Lemmas \ref{k_2},\ref{k_3} and \ref{c_l_40__k_geq_4}. To prove lemma \ref{c_l_40__k_geq_4}, we use lemma \ref{quet_binom} due to Anderson and Samuels, which bounds the difference between the binomial and Poisson distributions, and we use it to bound the approximation of the probability of getting a price $d$, $h_d$ for a small number of items.
We conclude by proving the lower bound of Theorem \ref{main_theorem} for any one-dimensional distribution.

\begin{lemma} \label{c_geq_40}
For any integer $k\geq 1$ and $c\geq 40$, the distribution $F$ over \(\{ 0,1\}\) with \(\mathbb{P}(1) = c/k\) satisfies
\[BRev(F^{k})/SRev(F^{k}) > 0.56\]
\end{lemma} 
\begin{proof}
Using the Chernoff bound we can bound the probability of selling in a bundle for price of at least  $0.65c$ for $c\geq 40$ by $1-\left(\frac{e^{-0.35}}{(0.65)^{0.65}}\right)^{40}>0.9$. 
For $c\geq 40$, let $d$ be an integer s.t. $0.625c \leq d \leq 0.65c$. Such $d$ always exists. Since $d\geq 0.625c$ and the probability of getting at least $d$ is more than $90\%$, $BRev(F^{k})/SRev(F^{k}) \geq 0.5625$.
\end{proof}

To prove the case of $c<40$ we preform a separate analysis of the cases $k=2,3$, and $k \geq 4$.

\begin{lemma} \label{k_2}
For $k=2$, and any $c$, the distribution $F$ over \(\{ 0,1\}\) with \(\mathbb{P}(1) = c/2\) satisfies
\[BRev(F^{2})/SRev(F^{2}) \geq 2/3\]
\end{lemma} 
\begin{proof}
We consider two cases. For $c\leq 4/3$, if we sell at price $1$ then the revenue ratio is $\frac{2(c/2)-(c/2)^2}{c}$, which in this range gives a ratio of at least $2/3$. For $c>4/3$, if we sell at price $2$ then the revenue ratio is $\frac{2((c/2)^2)}{c}$, which in this range also gives a ratio of at least $2/3$.
\end{proof}

\begin{lemma} \label{k_3}
For $k=3$, and any $c$, the distribution $F$ over \(\{ 0,1\}\) with \(\mathbb{P}(1) = c/3\) satisfies
\[BRev(F^{3})/SRev(F^{3}) > 0.6\]
\end{lemma} 
\begin{proof}
We consider three cases. For $c\leq 1.325$, if we sell at price $1$ then the revenue ratio is $\frac{1-(1-c/3)^3}{c}$, which in this range gives a ratio of at least $0.623$. For $1.325<c\leq 2.571$, if we sell at price $2$ then the revenue ratio is $\frac{2(1-(1-c/3)^3-3(c/3)(1-c/3)^2)}{c}$, which in this range also gives a ratio of at least $0.623$. For $c>2.5714$ if we sell at price $3$ then the revenue ratio is $\frac{3(c/3)^3}{c}$, which in this range gives a ratio of at least $0.734$.
\end{proof}

In order to conclude the lower bound for the case of $c<40$ and an arbitrary $k$, we use the following result by Anderson and Samuels \cite{anderson1965some} regarding the  Poisson approximation, $P(m;c)=e^{-c}$, to the binomial distribution, $B(m;k,c/k)$.

\begin{lemma}{(Anderson and Samuels \cite{anderson1965some})}  \label{quet_binom} 
The difference $P(m;c)-B(m;k,c/k)$ is positive when $m\leq c(1-1/(k+1))$.
 
\end{lemma}
From lemma \ref{quet_binom}, we conclude:
\begin{corollary} \label{cor_binom_poi}
If $d\leq c(1-1/(k+1))+1$, then $h_d(c)$ is a lower bound for the probability of preforming a sell when when offering a bundle price of $d$. 
\end{corollary}

Using this approximation we can handle the case of small $c$ by simply finding the optimal price for each $c$.
\begin{lemma} \label{c_l_40__k_geq_4}
For integer $k\geq 4$ and $c < 40$, the distribution $F$ over \(\{ 0,1\}\) with \(\mathbb{P}(1) = c/k\) satisfies
\[BRev(F^{k})/SRev(F^{k}) \geq \rr\]
\end{lemma} 
\begin{proof}
Given the parameter $c$ we can determent a bundle price which will give a revenue which reaches a bundle to separate ratio of at least the desired bound.
\begin{figure}
\centering
\includegraphics[width=0.5\textwidth]{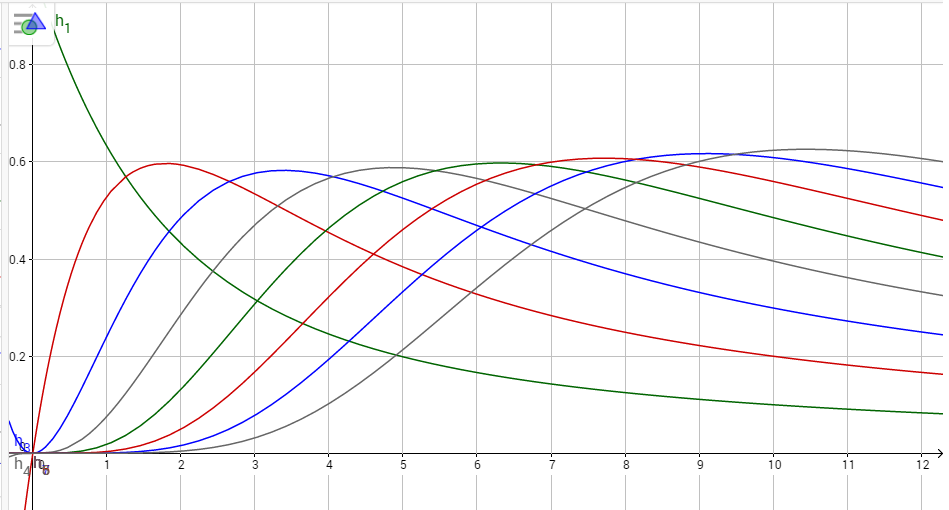}
\caption{$d\cdot h_d(c)/c$ for $d\leq 8$ as a function of $c$}
\end{figure}

The segments in left column of Table \ref{table:pricing_ta} are matched with a price which ensures a revenue which reaches at least the ratio in the right column for any $c<40$. The end points of each segment are where \( dh_d(x) =  (d+1)h_{d+1}(x)\). It is easy to see that the largest gap is indeed at $\cc$, the solution of \( 2h_2(x) =  3h_3(x)\).

\begin{table}[ht]
\centering
\begin{tabular}{ l | c || r }
  \hline
  $c$ segments & Price for segment & Gap in the end points\\
  $c\leq 1.256...$ & $d=1$ & 1, 0.569 \\
  $ 1.256... \leq c\leq 2.655...$ & $d=2$ & 0.569, \textbf{0.559} \\
  $ 2.655... \leq c\leq 4.057...$ & $d=3$ & \textbf{0.559}, 0.569 \\
  $ 4.057... \leq c\leq 5.441...$ & $d=4$ & 0.569, 0.581 \\
  $ 5.441... \leq c\leq 6.805...$ & $d=5$ & 0.581, 0.594 \\
  $ 6.805... \leq c\leq 8.152...$ & $d=6$ & 0.594, 0.605 \\
  $ 8.152... \leq c\leq 9.482...$ & $d=7$ & 0.605, 0.615 \\
  $ 9.482... \leq c\leq 10.79...$ & $d=8$ & 0.615, 0.624 \\
 $ 10.79... \leq c\leq 12.10...$ & $d=9$ & 0.624, 0.633 \\
 $ 12.10... \leq c\leq 13.39...$ & $d=10$ & 0.633, 0.641 \\
 $ 13.39... \leq c\leq 14.68...$ & $d=11$ & 0.641, 0.648 \\ 
 $ 14.68... \leq c\leq 15.95...$ & $d=12$ & 0.648, 0.654 \\
 $ 15.95... \leq c\leq 17.21...$ & $d=13$ & 0.654, 0.661 \\
 $ 17.21... \leq c\leq 18.47...$ & $d=14$ & 0.661, 0.666 \\ 
 $ 18.47... \leq c\leq 19.72...$ & $d=15$ & 0.666, 0.672 \\
 $ 19.72... \leq c\leq 20.96...$ & $d=16$ & 0.672, 0.677 \\
 $ 20.96... \leq c\leq 22.20...$ & $d=17$ & 0.677, 0.682 \\ 
 $ 22.20... \leq c\leq 23.44...$ & $d=18$ & 0.682, 0.686 \\
 $ 23.44... \leq c\leq 24.67...$ & $d=19$ & 0.686, 0.690 \\
 $ 24.67... \leq c\leq 25.89...$ & $d=20$ & 0.690, 0.694 \\ 
 $ 25.89... \leq c\leq 27.11...$ & $d=21$ & 0.695, 0.698 \\
 $ 27.11... \leq c\leq 28.33...$ & $d=22$ & 0.698, 0.702 \\
 $ 28.33... \leq c\leq 40      $ & $d=23$ & 0.702, 0.574 \\ 
 
  \hline  
\end{tabular}
\caption{Pricing for $c<40$}
\label{table:pricing_ta}
\end{table}

The derivative of $h_d(x)/x$ satisfies\( (h_d(x)/x)' = e^{-x} p(x)/x^2 + e^{-x} p(x)/x - 1/x^2 -e^{-x} p'(x)/x\) where $p(x)=\sum_{i = 0}^{d-1}{x^i/i!}$. For each $d$, the range dominated by $d$ has at most one extreme point in it which can easily be verified as a maximum point.
Thus, it is indeed enough to check only the end points of each range.\footnote{Both ends are needed due to some of the $d$s getting maximum points in the range dominated by $d$. See Figure 1.}
 
As stated by Corollary \ref{cor_binom_poi}, if $d\leq c(1-1/(k+1))+1$ then $h_d(c)$ is a lower bound for the probability of getting price $d$. For $k\geq 4$, the condition holds for all of the above.

\end{proof}
 
We can now prove the lower bound of Theorem \ref{main_theorem} for a general distribution.
\begin{theorem1}
For every integer $k\geq 1$ and every one-dimensional distribution $F$,
\[BRev(F^{k}) \geq \rr kRev(F) = \rr SRev(F^{k})\] 
\end{theorem1}
\begin{proof}
Let $X$ be distributed according to $F$; let $p$ be the optimal Myerson price for $X$ and $q = \mathbb{P}(X \geq p)$, so $Rev(F) = pq$.
Let $G$ be the distribution over \(\{ 0,p\}\) with \(\mathbb{P}(p) = q\). 

We claim that \(BRev(F^{k})/SRev(F^{k}) \geq BRev(G^{k})/SRev(G^{k})\). 

First, $SRev(G^{k}) = pqk = SRev(F^{k})$.
Second, for any price $pd$ where $d\in\mathbb{N}$, the probability of selling a bundle according to $F$ is at least as much as according to $G$, since in the latter case we must have at least $d$ events out of $k$ with
probability $q$ for each, and for $G$ it is sufficient for getting at least $d$. For any other price $gp$, the chances of selling a bundle according to $G$ is the same as selling in $\lfloor g\rfloor p$, and for $F$ the probability is at least that.

Let $c=qk$ and $H$ be the distribution over \(\{ 0,1\}\) with \(\mathbb{P}(1) = q = c/k\).
Since $BRev(G^{k}) = pBRev(H^{k})$, $SRev(G^{k}) = pSRev(H^{k})$ and for $H$ we have already shown this threshold, therefore the bound holds.
\end{proof}

\bibliography{a.bbl} 
\bibliographystyle{plain}

\end{document}